\documentclass[11pt,a4paper]{article}   

\usepackage{graphicx}

\usepackage{etoolbox} 
\usepackage[margin=12pt, font = small, labelfont =bf]{caption}
\usepackage{bm} 

\usepackage[english]{babel}
\usepackage{amsmath, amsthm,amssymb, multirow}
\usepackage{mathtools}
\usepackage{natbib}
\usepackage{enumerate}

\usepackage{tikz, ifthen}
\usetikzlibrary {shapes, trees, arrows, patterns, fadings}
\usetikzlibrary {positioning}
\usetikzlibrary{fit, backgrounds, calc, matrix}

\newtheorem{proposition}{Proposition}
\DeclareSymbolFont{bbold}{U}{bbold}{m}{n}
\DeclareSymbolFontAlphabet{\mathbbold}{bbold}
\DeclareBoldMathCommand{\bn}{n}
\DeclareBoldMathCommand{\bS}{S}
\DeclareBoldMathCommand{\bZ}{Z}
\DeclareBoldMathCommand{\bK}{K}
\DeclareBoldMathCommand{\bU}{U}

\newcommand{\ind}[1]{\ensuremath{\mathbbold{1}_{\left\{#1\right\}}}}

\DeclarePairedDelimiterX\cdp[2]{\lparen}{\rparen}{ #1\,\delimsize\vert\,#2}
\DeclarePairedDelimiterX\cdb[2]{\lbrace}{\rbrace}{ #1\,\delimsize\vert\,#2}
\newcommand{\cd}{\, |\, }
\newcommand\indep{\protect\mathpalette{\protect\independenT}{\perp}} 
\newcommand{\independenT}[2]{\mathrel{\rlap{$#1#2$}\mkern4mu{#1#2}}} 

\DeclareMathOperator{\Exp}{\mathbb{E}}
\DeclareMathOperator{\Var}{\mathbb{V}}
\DeclareMathOperator{\Pro}{\mathbb{P}}
\newcommand{\B}{\mathcal{B}}
\newcommand{\C}{\mathcal{C}}

\newcommand{\bin}[2]{\mathrm{bin}\left(#1, #2\right)}

\definecolor{oxfordblue}{RGB}{0,33,71}
\colorlet{backgroundcolor}{oxfordblue!10}
\pgfdeclarelayer{background}
\pgfdeclarelayer{layer1}
\pgfdeclarelayer{foreground}

\pgfsetlayers{background,layer1,main,foreground}
\newcommand{\myunit}{1em}

\newcommand{\convexpath}[2]{
[   
    create hullnodes/.code={
        \global\edef\namelist{#1}
        \foreach [count=\counter] \nodename in \namelist {
            \global\edef\numberofnodes{\counter}
            \node at (\nodename) [draw=none,name=hullnode\counter] {};
        }
        \node at (hullnode\numberofnodes) [name=hullnode0,draw=none] {};
        \pgfmathtruncatemacro\lastnumber{\numberofnodes+1}
        \node at (hullnode1) [name=hullnode\lastnumber,draw=none] {};
    },
    create hullnodes
]
($(hullnode1)!#2!-90:(hullnode0)$)
\foreach [
    evaluate=\currentnode as \previousnode using \currentnode-1,
    evaluate=\currentnode as \nextnode using \currentnode+1
    ] \currentnode in {1,...,\numberofnodes} {
-- ($(hullnode\currentnode)!#2!-90:(hullnode\previousnode)$)
  let \p1 = ($(hullnode\currentnode)!#2!-90:(hullnode\previousnode) - (hullnode\currentnode)$),
    \n1 = {atan2(\x1,\y1)},
    \p2 = ($(hullnode\currentnode)!#2!90:(hullnode\nextnode) - (hullnode\currentnode)$),
    \n2 = {atan2(\x2,\y2)},
    \n{delta} = {-Mod(\n1-\n2,360)}
  in 
    {arc [start angle=\n1, delta angle=\n{delta}, radius=#2]}
}
-- cycle
}

\newcommand{\clique}[4][2]{
  \let\mymatrixcontent\empty
  \foreach \c in {#4}{%
    \foreach[count = \j] \b in \c{
      \pgfmathparse{Mod(\j,#1)==0? 1:0}
      
      \ifthenelse{\pgfmathresult=1}{\expandafter\gappto\expandafter\mymatrixcontent\expandafter{\b\\}}
      {\expandafter\gappto\expandafter\mymatrixcontent\expandafter{\b\&}}
    }
    \expandafter\gappto\expandafter\mymatrixcontent\expandafter{[2]}
  }
  \matrix (#2) [matrix of math nodes, nodes in empty cells, ampersand replacement=\&, #3] {
    \mymatrixcontent
  };
}

\newcommand{\Bclique}[5][2]{
  \let\mymatrixcontent\empty
  \foreach \c in {#4}{%
    \foreach[count = \j] \b in \c{
      \pgfmathparse{Mod(\j,#1)==0? 1:0}
      \ifthenelse{\pgfmathresult=1}{\expandafter\gappto\expandafter\mymatrixcontent\expandafter{\b\\}}
      {\expandafter\gappto\expandafter\mymatrixcontent\expandafter{\b\&}}
    }
    \expandafter\gappto\expandafter\mymatrixcontent\expandafter{[2]}
  }
  \expandafter\gappto\expandafter\mymatrixcontent\expandafter{ 
    #5 \\
  }

  \matrix (#2) [matrix of math nodes, nodes in empty cells,  ampersand replacement=\&, #3] {
    \mymatrixcontent
  };
}

\tikzstyle{Cstyle}=[rounded corners = 0.4ex,draw,  font = \scriptsize, inner sep = 0.4ex,
nodes={text height=1.5ex,text depth=.25ex, inner sep = 0ex, rounded corners=0}, column sep = 0.5ex]

\newcommand{\hugin}{\texttt{HUGIN} }
\begin{document}

\title{Computational aspects of DNA mixture analysis\\
\vspace{10pt}
\small{Exact inference using auxiliary variables in a Bayesian network}}

\author{Therese Graversen\thanks{Corresponding author: Therese Graversen,
              Department of Statistics, University of Oxford,
              1 South Parks Road, Oxford OX1 3TG, United Kingdom,
              email: graversen@stats.ox.ac.uk. }\\University of Oxford       \and
        Steffen  Lauritzen\\ University of Oxford
}

\maketitle

\begin{abstract}
  Statistical analysis of DNA mixtures is known to pose computational
  challenges due to the enormous state space of possible DNA
  profiles. We propose a Bayesian network representation for
  genotypes, allowing
  computations to be performed locally involving only a few alleles at
  each step.  In addition, we describe a general method for computing the
  expectation of a product of discrete random variables using
  auxiliary variables and probability propagation in a Bayesian
  network, which in combination with the genotype network allows
  efficient computation of the likelihood function and various other
  quantities relevant to the inference. Lastly, we introduce a set of
  diagnostic tools for assessing the adequacy of the model for describing a
  particular dataset.\\
  
  \textbf{Keywords:} Bayesian network; genotype
    representation; junction tree; model diagnostics;
    prequential monitor; triangulation.
\end{abstract}

\section{Introduction}

In this paper we demonstrate methods for exact computation in
statistical analysis of DNA mixtures, where the need for summation
over the space of possible DNA profiles for unknown contributors
traditionally has involved some degree of approximation
\citep{article:PENDULUM,Tvedebrink2010,puch2012} and has only been made
for two or three unknown contributors \citep{Cowell2011202}. In
contrast, using the methodology presented here and the corresponding
implementation by \cite{graversen:package:13} in the
\texttt{R}-package $\texttt{DNAmixtures}$, \cite{cowell:etal:13} were
able to perform exact evaluation and subsequent numerical maximisation
of the likelihood function for up to six unknown contributors.

The present paper develops a suite of tools for inference in the
statistical model described in \cite{cowell:etal:13} enabling
evaluation of the likelihood function, computation of posterior
probability of genotypes given a set of observed peak heights, and
assessment of model adequacy. We exploit introduction of auxiliary
variables combined with an efficient representation of the genotypes
as a Bayesian network. The implementation in
 \texttt{DNAmixtures} 
 interfaces the HUGIN API \citep{hugin:api:manual:2013} via
\texttt{RHugin} \citep{manual:RHugin}.

The plan of the paper is as follows: Section~\ref{sec:model} briefly
describes the relevant model for DNA mixture analysis and the
computational methods are detailed in \ Section~\ref{sec:bayesnet}. In
Section~\ref{sec:mixtureanal} we show how the methodology can be
extended to calculate various quantities of interest; in particular we
develop  methods for assessing the adequacy of the model.

\section{A statistical model for mixed traces of DNA}
\label{sec:model}
In statistical analysis of DNA mixtures it is of interest to draw
inference about individual DNA profiles in a mixed trace of DNA. The observations consists of a set of peak heights in an electropherogram (EPG) produced after a che\-mi\-cal duplication process known as a polymerase chain reaction (PCR). 
Figure~\ref{fig:epg} represents a schematic illustration of part of an EPG.
\begin{figure}
  \centering
  \includegraphics{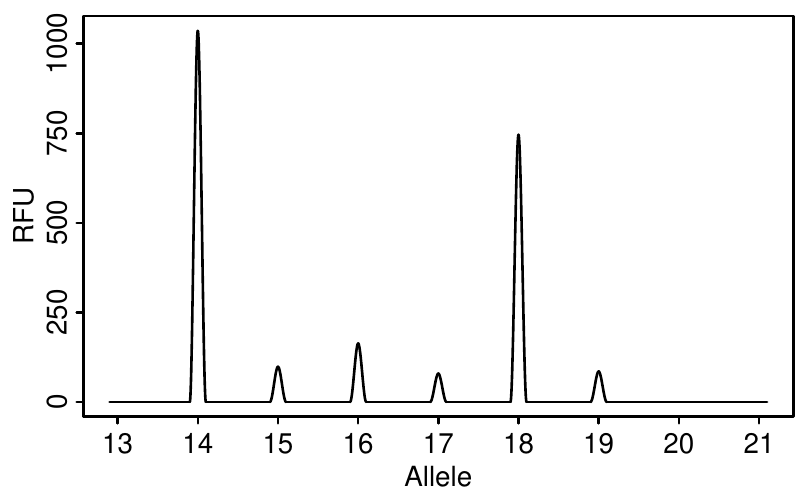}    
  \caption{Stylized electropherogram exhibiting peaks for the alleles at one particular marker.}
  \label{fig:epg}
\end{figure}

The DNA sequences  at Short Tandem Repeat (STR) \emph{mar\-kers}
are characterised by a motif of base pairs repeated a number of times,
so that a specific repeat number corresponding to an \emph{allele} and
a peak in the EPG typically indicates presence of the corresponding
allele.

A pair of alleles is called a \emph{genotype}, and the genotypes
across a set of markers constitute the \emph{DNA profile}.  The
markers used for forensic identification are typically located at
different chromosome-pairs or, if not, on well separated locations,
rendering it reasonable to assume independence of genotypes across
markers.

The observed EPG is prone to artefacts known as \emph{stutter} and
\emph{dropout}: Stutter refers to the phenomenon that some of the DNA
alleles may lose a repeat motif during the PCR process and thus
contribute to a peak at a lower repeat number. If there are allelic
types $\{1, \ldots, A\}$ we thus assume that any allele $a$ receives
stutter from the amplification of allele $a+1$. Dropout refers to the
fact that peak heights occasionally are too small for the peak to be
registered.

We distinguish between known and unknown contributors to the sample,
depending on whether their DNA profile is considered known or not.
The computational complexity of the problem is directly associated
with the huge number of
possible allocations of genotypes to the unknown contributors.

The genotypes of a DNA profile are assumed independent across markers
so we briefly describe the model for one marker only, following
\cite{cowell:etal:13}.

\subsection{Model for the genotypes of unknown contributors}

We assume that the alleles of an unknown person are sampled from a
reference population in Hardy--Weinberg equilibrium so that the two
alleles can be considered sampled independently.  Denote by $n_{ia}$
the number of alleles of type $a$ for contributor $i$. A genotype
$(n_{i1}, \ldots, n_{iA})$ for an unknown contributor  follows a
multinomial distribution with  allele frequencies
$(q_{1}, \ldots, q_{A})$ and $\sum_a n_{ia}=2$.  Unknown contributors are assumed unrelated
so their genotypes are independent.

\subsection{Peak height distribution for fixed genotypes}

Analysing DNA that contains alleles of type $a$ results in a peak at
position $a$, and possibly also a smaller peak at position $a-1$ due
to stutter during the PCR process; thus the height of the peak
$H_a\geq 0$ for allele $a$ depends on the presence of alleles of both
type $a$ and type $a+1$.  Peaks of height $H_a$ below a chosen
threshold $C$ are not registered and so the observed peak heights are
$Z_a = H_a \ind{H_a \ge C}$.

For given genotypes of the contributors we assume that 
the peak height $H_a$ at allelic type $a$ is gamma
distributed 
with  shape and scale parameters depending on the numbers
$n_{ia}$, $n_{i,a+1}$ of alleles of type $a$ and $a+1$ that each
unknown contributor $i$ possesses, as well as  a
set of model parameters, $\psi = (\rho, \eta, \xi, \phi)$; more precisely we assume that 
$H_a \sim \Gamma(\lambda_a, \eta) $,
where
\begin{equation}\label{eq:gammadist}
\lambda_a=
\rho\sum_{i = 1}^k \left\{(1-\xi)n_{ia} + \xi n_{i,a+1}\right\}\phi_i.
\end{equation} 
Here, and in the following, we have let $n_{i,A+1}=0$; the parameter $\xi$ is the \emph{mean stutter percentage}, $\rho$ is related to the general peak variability, $\phi_i$ denotes the fraction of DNA from individual $i$, and $\eta$ is the scale.  If $\lambda_a=0$ the gamma distribution $\Gamma(0,\eta)$ is considered degenerate at 0.
\subsection{Likelihood function}
The likelihood function is determined by the distribution of the observed peak heights. 
The observed peak heights are independent across markers $m=1, \ldots, M$, and thus the likelihood
function factorises accordingly. Using this fact in combination with \eqref{eq:gammadist} we find 
\begin{align}
  \ell(\psi)
  & = \prod_{m = 1}^Mf_\psi(Z_1^m, \ldots, Z_{A_m}^m) \nonumber \\
  & = \prod_{m = 1}^M \Exp\left\{f_\psi\cdp*{Z_1^m, \ldots, Z_{A_m}^m}{\bn}\right\} \nonumber  \\
  &= \prod_{m=1}^M\Exp \bigg\{\prod_{a=1}^{A_m} f_\psi\cdp{z^m_a}{\bn_a, \bn_{a+1}}\bigg\}, \label{eq:likfunction}
  \end{align}
  where the expectation is taken with respect to the distribution of genotypes of the unknown contributors.
Here and in the following $\bn$ denotes the full set of genotypes for all individuals and $\bn_a$ the vector $\bn_a= (n_{ia}, i\in I)$ of allele-counts of type $a$. The expectation in \eqref{eq:likfunction} involves summation over all combinations of possible genotypes of  potential contributors. 
There are $\{A_m(A_m+1)/2\}^k$ possible combinations of genotypes at
a marker, and thus there are this many terms in the sum, each being a product of $A_m$ factors. Direct computation is  typically infeasible when there are many alleles and many unknown contributors. We attack this computational problem by appropriate use of Bayesian network techniques, as detailed in Section~\ref{sec:bayesnet} below.

Note that our methodology can be used directly with other choices of distribution
for the peak heights, provided that the distribution of the peak
height for allele $a$ depends only on the genotypes through the number
of alleles of types $a$ and $a+1$.

\section{Computational methods}\label{sec:bayesnet}

As a consequence of \eqref{eq:likfunction}, and for other purposes,
the computational task in DNA mixture analysis involves repeated
computation of the expectation $\Exp\{h(X)\}$
of non-negative functions $h$ of a set
of discrete variables $X = \{X_v\}_{v \in V}$.
We describe our computational approach in thid general setting
before returning to the DNA mixture model in Section~\ref{sec:gtnet},
where we give a network representation of a genotype for an unknown
contributor to the trace.

\subsection{Computation by auxiliary variables}\label{sec:dummy}

Let $X = \{X_v\}_{v\in V}$ be a collection of discrete variables with
a distribution represented by a Bayesian network.  For $B \subseteq
V$, we denote by $X_B$ the collection of variables $\{X_v\}_{v\in B}$.

Let $h$ be a non-negative function which can be written on the form
\[
h(x) = \prod_{B \in \B}h_B(x_B),
\]
for some set $\B$ of subsets of $V$ and 
real-valued, non-negative functions $h_B$.

For each $B\in\B$ we introduce binary random variables $Y^B\in \{0,1\}$
which are conditionally independent given the network and have conditional distributions 
\begin{equation}\label{eq:dummydef}
\Pro \cdp*{Y^B = 1}{X = x} =
\Pro \cdp*{Y^B = 1}{X_B = x_B} 
= h_B(x_B)/k^B.
\end{equation}
Here, the constant $k^B$ is chosen such that $h_B(x_B)/k^B \in [0, 1]$
over all states $x_B$ and so \eqref{eq:dummydef} defines a valid probability distribution. A simple choice
would be $k^B = \max_{x_B}{h_B(x_B)}$, i.e.\ the largest value that
$h_B$ attains over the state space of $X_B$. 
We use the state space $\{0,1\}$ for auxiliary variables, but note that
this choice is unimportant for the method itself.

The desired expectation $\Exp\{\prod_{B\in\B}h_B(X_B)\}$ can now be
expressed as the probability of a specific configuration of the binary variables
introduced. As Proposition~\ref{the:exptoprob} reveals, this is also the
case for the expectation of a product of any subset of the variables $h_B(X_B)$.

\begin{proposition}\label{the:exptoprob}
  For all $\B' \subseteq \B$ it holds that
  \[
  \Exp\bigg\{\prod_{B\in\B'}h_B(X_B)\bigg\} = \Pro\bigg(\bigcap_{B\in\B'}\{Y^B=1\}\bigg)\prod_{B\in\B'}k_B. 
  \]    
\end{proposition}

\begin{proof}
  Using (\ref{eq:dummydef}) and the fact that $Y^B$  are conditionally
  independent given $X$ we get
  \begin{alignat*}{3}
    \Exp\bigg\{\prod_{B\in\B'}h_B(X_B)\bigg\} &= \Exp\bigg\{\prod_{B\in\B'}\Big(\Pro\cdp*{Y^B=1}{X_B}k_B\Big)\bigg\}&\\
    &= \Exp\bigg\{\prod_{B\in\B'}\Pro\cdp*{Y^B=1}{X}\bigg\}\prod_{B\in\B'}k_B&\\
    &= \Exp\bigg\{\Pro\cdp[\Big]{\bigcap_{B\in\B'}\{Y^B=1\}}{X}\bigg\}\prod_{B\in\B'}k_B&\\
    &= \Pro\bigg(\bigcap_{B\in\B'}\{Y^B=1\}\bigg)\prod_{B\in\B'}k_B
  \end{alignat*}
  as desired. %
\end{proof}

If the distribution of the variables $\{X_v\}_{v\in V}$ is modelled by a 
Bayesian network, this network can be extended to
include the  variables $\{Y^B\}_{B \in \B}$ by for each $B$
adding $Y^B$ as a child of $\{X_v\}_{v \in B}$ with conditional
distributions of $Y^B$ in \eqref{eq:dummydef}.
As the auxiliary variables are added as children of existing network nodes, no
directed cycles are created and the extended network is a correct representation of the joint
distribution of $(X,Y)$  since, given $X_B$, $Y^B$ is conditionally independent of all other variables in the extended network.

Figure~\ref{fig:bnexample} illustrates how the network is extended in
case of a function $h$ factorising over two sets of variables $(
X_2, X_3)$ and $(X_3, X_4, X_5)$.
\begin{figure}[ht]
  \centering
\begin{tikzpicture}[node distance = 2.5em and 2em, font = \small,
  every node/.style = {draw, ellipse, fill = white}, every path/.style = {->, >=latex'}]
  \node   at (0,0) (x1){$X_1$};
  \node   [below right = 1em and 2em of x1](x3){$X_3$};
  \node   [below left = of x1](x2){$X_2$};
  \node   [right = 6em of x1](x4){$X_4$};
  \node   [below right = of x4](x5){$X_5$};
  \draw   (x1) -- (x2);
  \draw   (x3) -- (x4);
  \draw   (x1) -- (x4);
  \draw   (x4) -- (x5);
  \draw   (x1) -- (x3);
  \begin{pgfonlayer}{background}
    \draw [fill = backgroundcolor!50]\convexpath{ x3, x2}{1.5em};
    \draw [fill = backgroundcolor!50]\convexpath{x3, x4, x5}{1.5em};
    \draw [draw = black, fill = none]\convexpath{x3, x2}{1.5em};
    \draw [draw = black, fill = none]\convexpath{x3, x4, x5}{1.5em};    
  \end{pgfonlayer}

  \node [below = 6em of x1, fill = backgroundcolor](y1){$Y^{\{2,3\}}$};
  \node  [below = 6em of x4, fill = backgroundcolor](y2){$Y^{\{3,4,5\}}$};
  \draw  [dashed](x4) -- (y2);
  \draw  [dashed](x5) -- (y2);
  \draw  [dashed](x3) -- (y2);
  \draw  [dashed](x2) -- (y1);
  \draw  [dashed](x3) -- (y1);
\end{tikzpicture}
  
\caption{Extending a network with two binary variables for computation of
  $\Exp\left(h_{\{2,3\}}(X_2, X_3)h_{\{3,4,5\}}(X_3, X_4,
    X_5)\right)$. Here $\B = \{\{2, 3\},\{3, 4, 5\}\}$}
  \label{fig:bnexample}
\end{figure}
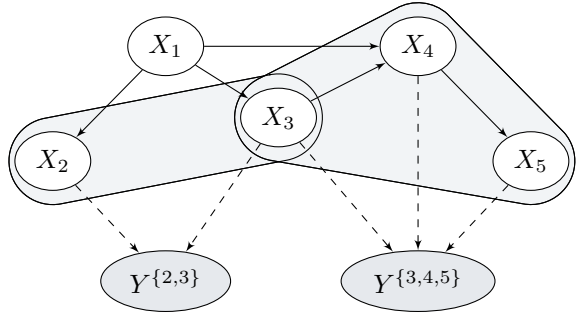

\subsubsection{Probability propagation}

We now briefly describe probability propagation and explain how to exploit the normalising
constants arising as a by-product of the propagation algorithm. 
We refer for example to \cite{cowell:etal:99} for further details.

A computational structure is set up in the form of a so-called \emph{junction tree} of subsets of the variables involved: first an undirected graph, the \emph{mo\-ra\-lised graph},  is constructed by adding undirected links between nodes that have a common child and removing directions for existing edges.  Subsequently links are added to ensure that the resulting graph is chordal. This process is known as \emph{triangulation} and can generally be done in many ways. Finally the cliques in the triangulated graph are arranged in a junction tree. 

In the situation described above $X_B$ is the parent set of $Y^B$ in
the extended network and the node set $X_B$ will  thus be a complete set in the triangulated graph, hence contained in  some clique. The efficiency of the method depends crucially on the size of cliques for the chosen triangulation, see further discussion in Section~\ref{sec:complexity} below.

A distribution $p(x)$ is represented by an unnormalised
probability function \[p(x)\propto g(x)=\frac{\prod_{C\in \C}\zeta_C(x_C)}{\prod_{S\in {\mathcal S}}\zeta_S(x_S)}\] where $\mathcal S$ denotes the set of  \emph{separators}, i.e.\ intersections of pairs of neighbouring cliques in the junction tree. The corresponding
normalising constant  is $N_1 = \sum_x g(x)$. The function $g(x)$ is known as the \emph{charge} and the functions $\zeta$ as \emph{potentials}.

A message passing operation referred to  as 
\emph{propagation} brings the charge on a canonical form, where all
potentials of the charge are equal to the function  $g$ marginalised onto the
corresponding clique or separator, i.e.\
\[\zeta_D(x_D)= \sum_{y: y_D= x_D} g(y) \mbox{ for all $D\in \C\cup \mathcal{S}$.}\]
The normalising constant can then be computed efficiently after propagation
as $\sum_{x_D}\zeta_D(x_D)$, for
instance choosing $D$ as a separator $S\in \mathcal{S}$ with minimal state space.

The charge $g$ can be modified by entering so-called \emph{likelihood evidence} $\ell_v(x_v)$ on single nodes leading to the charge 
\[\tilde g(x) = g(x)\prod_{v \in V}\ell_v(x_v)
\]
with normalising constant
\[
N_2 = \sum_{x} g(x)\prod_{v \in V}\ell_v(x_v).
\]

Taking the ratio of the normalising constants before and after
propagating the likelihood evidence yields the expectation of the
product of the likelihood evidence with respect to the distribution
$p(x)$:
\begin{align*}
  \frac{N_2}{N_1} &= \frac{\sum_x g(x)\prod_{v\in V} \ell_v(x_v)}{\sum_y g(y)}
  = \sum_x \frac{g(x)}{\sum_y g(y)}\prod_{v\in V} \ell_v(x_v)\\
  &=\sum_x  p(x)\prod_{v\in V} \ell_v(x_v)
  = \Exp \bigg\{\prod_{v\in V} \ell_v(X_v)\bigg\}.
\end{align*}

As shown in Proposition~\ref{the:nc} below, this fact now ensures that the expectation of interest can be calculated by propagating likelihood evidence on the auxiliary variables.
\begin{proposition}\label{the:nc}
 Let likelihood evidence for each node $Y^B$, $B \in \B' \subseteq \B$ be given as:
  \[
  \ell_B(Y^B) = \begin{cases}
    k_B, &Y^B = 1\\
    0, &Y^B = 0\\
  \end{cases}
  \]
  and let $N_1$ and $N_2$ be the normalising constants before and
  after propagation of the likelihood evidence. Then we have
  \[
  \Exp\bigg\{\prod_{B\in \B'}h_B(X_B)\bigg\} = \frac{N_2}{N_1}.
  \]            
\end{proposition}

\begin{proof}
  \begin{align*}
    \frac{N_2}{N_1} &= \Exp\bigg\{\prod_{B \in \B}\ell_B(Y^B)\bigg\} \\
    &= \Exp\bigg(\prod_{B \in \B'}k_B \ind{Y^B = 1}\bigg)\\
    &=  \Pro\bigg(\bigcap_{B \in\B'}\{Y^B = 1\}\bigg)\prod_{B \in \B'}k_B
  \end{align*}
  which by Proposition~\ref{the:exptoprob} equals the desired
  expectation. 
\end{proof}
\subsection{A Bayesian network representation of genotypes}\label{sec:gtnet}

The multinomial distribution of allele-counts $(n_{i1},\dots, n_{iA})$ representing the genotype of individual $i$ does not  in itself have  Markovian properties. However, if we define the partial sums $S_{ia} = \sum_{b = 1}^a n_{ia}$ counting the number of
alleles of type up to and including $a$ that person $i$ possesses, we can represent the genotype in a Bayesian network as displayed in Figure~\ref{fig:onegt}.

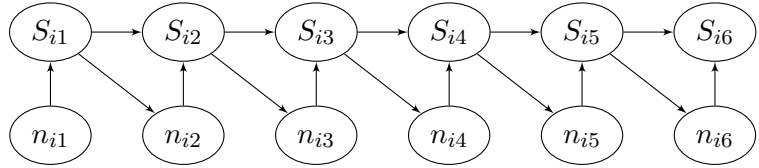
\begin{figure}[htb]
  \centering
  \begin{tikzpicture}[every path/.style={->, >=latex'}]    
  \tikzstyle{BNnode}=[shape = ellipse, fill = white, draw, font = \normalsize, text height=1.5ex,text depth=.25ex]
  \tikzstyle{GTmat}=[matrix of math nodes,  execute at empty cell={\node [draw = none, fill = none]{};},
  ampersand replacement=\&, nodes = {BNnode}, row sep = 1.5\myunit, column sep = 1.7\myunit]
  \tikzstyle{Oarrow}=[]

  \matrix (gt) [GTmat]{
    S_{i1} \& S_{i2} \& S_{i3} \& S_{i4} \& S_{i5} \& S_{i6} \\
    n_{i1} \& n_{i2} \& n_{i3} \& n_{i4} \& n_{i5} \& n_{i6} \\
  };
  
  \foreach \this [evaluate = \this as \prev using int(\this-1)] in {1, ...,6}
  {
    \draw (gt-2-\this) -- (gt-1-\this); 
    
    \ifthenelse{\NOT 1 = \this}{ 
      \draw (gt-1-\prev) -- (gt-1-\this); 
      \draw (gt-1-\prev) -- (gt-2-\this); 
    }
    {};
  }
\end{tikzpicture}  

  \caption{Network representation of a genotype at a marker with $A=6$ allelic types.}
  \label{fig:onegt}
\end{figure}

If we imagine the two alleles in the genotype being
allocated sequentially, then the number of alleles that a person has
of type $a+1$ only depends on how many alleles of the total two are left
to allocate, and the allocation happens according to a binomial
distribution. In Proposition~\ref{prop:multinomial} we establish the formal correctness of the network specification. 

\begin{proposition}\label{prop:multinomial} The distributions of genotypes and partial sums satisfy the following relations
  \begin{align}  
  S_{i1}&= n_{i1}, \nonumber \\
  n_{i1} &\sim \bin{2}{q_1}, \nonumber
  \intertext{and for $\ a \in \{2, \ldots, A\}$}
   S_{ia} &= S_{i,a-1} + n_{ia}, \nonumber\\
  n_{ia} \cd S_{i,a-1} &\sim \bin{2-S_{i,a-1}}{q_{a}/\textstyle{\sum_{b=a}^A q_b}}.
  \label{eq:condbin}
    \end{align}
    Finally, we have the conditional independence relations
    \begin{align}
    n_{ia} &\indep (n_{i1}, \ldots, n_{i,a-1}, S_{i1}, \ldots, S_{i,a-2}) \cd S_{i,a-1} \label{eq:condindep}\\
    S_{ia} &\indep (n_{i1}, \ldots, n_{i,a-1}, S_{i1}, \ldots, S_{i,a-2}) \cd (S_{i,a-1}, n_{ia}). \nonumber
  \end{align}
\end{proposition}
\begin{proof}The unnumbered relations follow directly from the definition of the quantities involved. We further have
\begin{eqnarray*}
\lefteqn{p(n_{ia}\cd n_{i1}, \ldots, n_{i,a-1})=
\frac{p(n_{i1}, \ldots, n_{i,a-1},n_{ia})}{p(n_{i1}, \ldots, n_{i,a-1})}}\\&=&
\frac{
\frac{2!}{(2- S_{i,a-1}-n_{ia})!\prod_{b=1}^{a}n_{ib}!}
 \left( \sum_{b=a+1}^{A}q_b\right)^{2- S_{i,a-1}-n_{ia}}{\prod_{b = 1}^a q_b^{n_{ib}}}}
{\frac{2!}{(2- S_{i,a-1})!\prod_{b=1}^{a-1}n_{ib}!} \left( \sum_{b=a}^{A}q_b\right)^{2- S_{i,a-1}}\prod_{b = 1}^{a-1} q_b^{n_{ib}}}\\&=&
\frac{(2-S_{i,a-1})!}{n_{ia}!(2- S_{i,a-1}-n_{ia})!}\\&&\times
\left( 1- \frac{q_a}{\sum_{b=a}^{A}q_b}\right)^{2- S_{i,a-1}-n_{ia}}\left(\frac{q_a}{\sum_{b=a}^{A}q_b}\right)^{n_{ia}}.
\end{eqnarray*}
The conditional independence \eqref{eq:condindep} follows from the fact that the conditional distribution of $n_{ia}$ given $n_{i1}, \ldots, n_{i,a-1}$ only depends on the condition through $S_{i,a-1}$; inspection of the expression for the conditional distribution yields \eqref{eq:condbin}. 
\end{proof}

\subsection{Auxiliary variables for computing the likelihood
  function}\label{sec:likelihood}

In order to compute the inner expectation in \eqref{eq:likfunction}, we note
that this is an expectation of a product over alleles, where each
factor is a function of the variables $\bn_a$ and $\bn_{a+1}$, and so
we can compute this expectation using auxiliary variables as described in Section~\ref{sec:dummy}:
For each allele $a$, we add an auxiliary variable $O_a$ with parents
$n_{ia}$ and $n_{i,a+1}$ for all unknown contributors $i$, except for
$O_A$ that is given only one parent $n_{iA}$ per contributor. Figure~\ref{fig:gtnet} shows the network for modelling one marker of a
mixture with two contributors and six alleles. 
\begin{figure}[htb]
  \centering
  \begin{tikzpicture}[every path/.style={->, >=latex'}]
    
  \tikzstyle{BNnode}=[shape = ellipse, fill = white, draw, font = \small, text height=1.5ex,text depth=.25ex]
  \tikzstyle{GTmat}=[matrix of math nodes,  execute at empty cell={\node [draw = none, fill = none]{};},
  ampersand replacement=\&, nodes = {BNnode}, row sep = 1.5\myunit, column sep = 1.7\myunit]
  \tikzstyle{Oarrow}=[]

  \matrix (gt) [GTmat]{
    S_{i1} \& S_{i2} \& S_{i3} \& S_{i4} \& S_{i5} \& S_{i6} \\
    n_{i1} \& n_{i2} \& n_{i3} \& n_{i4} \& n_{i5} \& n_{i6} \\
    [+1.5\myunit]
    O_1 \& O_2 \& O_3 \& O_4 \& O_5 \& O_6 \\
    [+1.5\myunit]
    n_{j1} \& n_{j2} \& n_{j3} \& n_{j4} \& n_{j5} \& n_{j6} \\
    S_{j1} \& S_{j2} \& S_{j3} \& S_{j4} \& S_{j5} \& S_{j6} \\
  };
  
  \foreach \this [evaluate = \this as \prev using int(\this-1)] in {1, ...,6}
  {
    \draw (gt-2-\this) -- (gt-1-\this); 
    \draw (gt-4-\this) -- (gt-5-\this); 
    \draw [Oarrow](gt-2-\this) -- (gt-3-\this); 
    \draw [Oarrow](gt-4-\this) -- (gt-3-\this); 
    
    \ifthenelse{\NOT 1 = \this}{ 
      \draw (gt-1-\prev) -- (gt-1-\this); 
      \draw (gt-1-\prev) -- (gt-2-\this); 
      \draw (gt-5-\prev) -- (gt-5-\this); 
      \draw (gt-5-\prev) -- (gt-4-\this); 
      \draw [Oarrow](gt-2-\this) -- (gt-3-\prev); 
      \draw [Oarrow](gt-4-\this) -- (gt-3-\prev); 
    }
    {};
  }
  \begin{pgfonlayer}{background}
    \node[rectangle, rounded corners = \myunit, draw = none, fill = black!10, fit = (gt-1-1)(gt-2-6)](bg1){};
    \node[rectangle, rounded corners = \myunit, draw = none, fill = black!10, fit = (gt-4-1)(gt-5-6)](bg2){};
  \end{pgfonlayer}
\end{tikzpicture}  
\caption{Bayesian network modelling the genotypes of 2 unknown
  contributors $i$ and $j$ for a marker with 6 possible allelic types.}
  \label{fig:gtnet}
\end{figure}
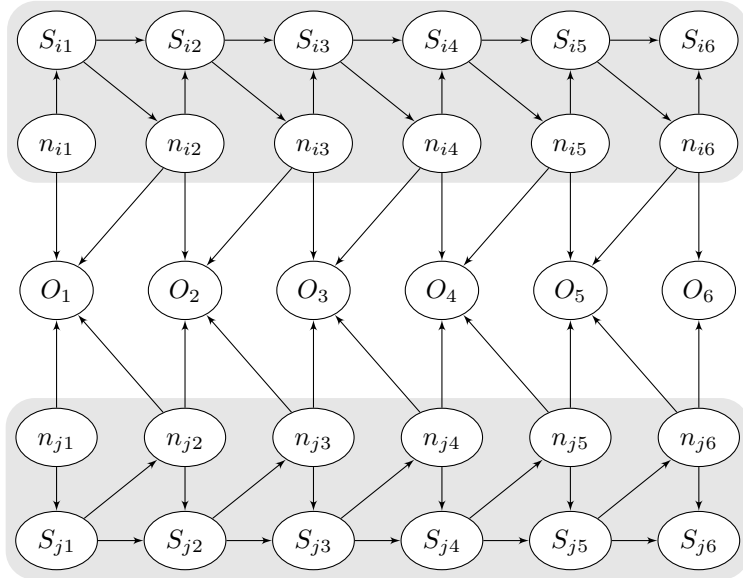
Note that $O_a$ and its parents
$n_{ia}, n_{i,a+1}$, $i \in \{1, \ldots, k\}$ are necessarily contained in the
same clique, implying that any valid junction tree will contain cliques with an associated state space that is
exponential in the number $k$ of unknown contributors. Unfortunately, as
the moralised graph is not chordal -- for instance $(S_{i1},n_{i1},
n_{j2}, n_{i3}, S_{i2}, S_{i1})$ is a cycle -- further edges need to
be added, resulting in an additional increase in the size of the cliques. We
shall return to this issue in Section~\ref{sec:complexity}.

The distribution of a peak height $Z_a$ conditionally on the allele-counts is for $a<A$
\begin{equation}
  f_\psi\cdp{z_a}{\bn_a, \bn_{a+1}}
  =
  \begin{cases}
    g_\psi\cdp*{z_a}{\bn_a, \bn_{a+1}}, & z_a \geq C\\
    G_\psi\cdp*{C}{\bn_a, \bn_{a+1}}, & z_a < C
  \end{cases}\label{eq:peakdensity}
\end{equation}
where $g$ and $G$ denotes the density respectively the cumulative
distribution function for the gamma distribution with parameters as in \eqref{eq:gammadist}. 

Define the distribution of $O_a$ for an observed allele, where $z_a \ge
C$, as
\begin{equation}\label{eq:condobs}
  \Pro\cdp{O_a = 1}{\bn_a, \bn_{a+1}} = g_\psi\cdp{z_a}{\bn_a, \bn_{a+1}}/k^\psi_a,
\end{equation}
noting the dependence of the scaling factor $k^\psi_a$
on $\psi$. 
For an unobserved allele, where $z^m_a = 0$, let the distribution of $O_a$ be defined as
\begin{equation}\label{eq:condunobs}
  \Pro\cdp{O_a = 0}{\bn_a, \bn_{a+1}} = G_\psi\cdp{C}{\bn_a, \bn_{a+1}}.
\end{equation}
For convenience we have defined the auxiliary variables so that for all
alleles the event $O_a = 1$ corresponds to the event that the peak at
allele $a$ is above the  threshold $C$.

Now Proposition~\ref{the:nc} can readily be used to evaluate the
contribution to the likelihood from marker $m$ for a given value of $\psi$ by
propagating likelihood evidence 
\begin{equation}
  \label{eq:likevidence}
  \ell_a(O_a) =
  \begin{cases}
    k^\psi_a \ind{O_a = 1}, &\textrm{if $a$ is seen}\\
    \ind{O_a = 0}, &\textrm{if $a$ is unseen.}    
  \end{cases}  
\end{equation}

\subsection{Posterior distribution of genotypes}\label{sec:posterior}
 When entering and propagating likelihood evidence as in \eqref{eq:likevidence} for $a$ in a set of alleles $B$, we obtain a representation of the conditional distribution of the full network given the relevant state of the auxiliary variables $O_a, a\in B$. Furthermore, 
this distribution is identical to the conditional distribution of the nodes in the network given the peak
height information $\{z_a\}_{a\in B}$:
\begin{equation} p\cdp*{x}{\{z_a\}_{a\in B}} = p\cdp[\Big]{x}{\bigcap_{\stackrel{a\in B,}{z_a > C}} \{O_a = 1\}\bigcap_{\stackrel{a\in B,}{z_a = 0}} \{O_a = 0\}}
\label{eq:evidence}
\end{equation}
 This follows from the following argument:
\begin{eqnarray*}
\lefteqn{p(x)\prod_{a\in B}\ell_a(O_a)}\\
  &\propto&p\cdp[\Big]{x}{\bigcap_{\stackrel{a\in B,}{z_a > C}} \{O_a = 1\}\bigcap_{\stackrel{a\in B,}{z_a = 0}} \{O_a = 0\}} \nonumber \\
  &\propto& p(x)\Pro\cdp[\Big]{\bigcap_{\stackrel{a\in B,}{z_a > C}} \{O_a = 1\}\bigcap_{\stackrel{a\in B,}{z_a = 0}} \{O_a = 0\}}{x}\nonumber \\
  & =&  p(x)\prod_{\stackrel{a\in B}{z_a > C}}\Pro\cdp*{O_a = 1}{x}\prod_{\stackrel{a\in B}{z_a = 0}}\Pro\cdp*{O_a = 0}{x}\nonumber \\
  & =&  p(x)\prod_{\stackrel{a\in B}{z_a > C}}\Pro\cdp*{O_a = 1}{\bn_{a}, \bn_{a+1}}\prod_{\stackrel{a\in B}{z_a = 0}}\Pro\cdp*{O_a = 0}{\bn_{a}, \bn_{a+1}}\nonumber \\
  & =&  p(x)\prod_{\stackrel{a\in B}{z_a > C}}\{g_\psi\cdp{z_a}{\bn_a, \bn_{a+1}}/k^\psi_a\}\prod_{\stackrel{a\in B}{z_a = 0}}G_\psi\cdp{C}{\bn_a, \bn_{a+1}}\nonumber \\
  & \propto&  p(x)\prod_{a\in B} f_\psi\cdp*{\{z_a\}_{a\in B}}{x} \nonumber\\
  & \propto& p\cdp*{x}{\{z_a\}_{a\in B}}. 
\end{eqnarray*}

As a consequence,  we can easily sample from the conditional distribution of genotypes given peak height information, which we shall exploit  in Sections \ref{sec:simulation} and \ref{sec:map} below.

\subsection{Network complexity considerations}

The main concerns when applying computation by auxiliary variables to a specific problem are that the junction tree representation of the network may not fit in the physical memory, and propagation and other network operations may take prohibitively long. Both of these issues are directly related to the \emph{total size} of the network junction tree. 
An additional concern lies in finding a good triangulation, as this can be both time- and memory-consuming; we eliminate this additional cost by specifying triangulations directly.

The total size of the
junction tree is the sum of the sizes of state spaces for all
cliques and separators and  
determines how many numbers are needed to store the clique and
separator tables. 

Once a junction tree has been created for a network, computation by
auxiliary variables involves
setting the conditional probability tables for each auxiliary variable and propagating evidence.
The number of elementary arithmetic operations for propagation is linear in the total size. Also, the number of cells that need updating when the conditional probability tables for the auxiliary variables change is, in the worst case, determined by the total size. 

In the following we study the relation of the total sizes of junction tree representations  used for DNA mixture analysis to the number $A$ of possible alleles at a marker and the number $k$ of unknown contributors.

\subsubsection{Junction tree sizes for DNA mixtures}\label{sec:complexity}

We shall consider three different triangulations of networks of the type discussed in Section~\ref{sec:likelihood} and investigate the behaviour of the total sizes of the corresponding junction trees. 
We restrict attention to mixture networks 
where any allele $a$ --- apart from the last allele $A$ --- can receive stutter
from $a+1$.

Any triangulation must necessarily have cliques that contain  auxiliary variables with their parent sets as these are complete sets in the moralised graph.  For all our junction trees we avoid adding additional variables to all such sets and simply combine any auxiliary variable with its parent set to form a clique. We can thus focus the discussion on triangulating the part of the moralised graph that does not involve auxiliary variables.

If we have $N$ binary auxiliary variables per allele, their cliques and corresponding separators contribute to the total size of the junction tree by 
\[
TS_{\mathrm{aux}} = 3N\left\{(A-1)3^{2k} + 3^k\right\},
\]
since there are $N(A-1)$ cliques containing an auxiliary variable along with its $2k$ parents, and each is separated from the remaining junction tree by a separator containing the $2k$ parents. The $N$ auxiliary variables for the last allele have only $k$ parents.

Bearing Figure~\ref{fig:onegt} in mind, the structure of the genotype networks requires \emph{upper triangle} sets
$\{S_{i,a-1}, S_{ia}, n_{ia}\}$ to be in a clique as they are
complete sets. If allele $a-1$ receives stutter from $a$, then the
\emph{lower triangle} set $\{n_{i,a-1}, n_{ia}, S_{ia}\}$ is also complete in the moralised graph and must be
contained in some clique. 

The first triangulation method we shall consider, uses the simple idea of slicing the network into cliques
\[\{S_{ia}, S_{i,a+1}, n_{ia} , n_{i,a+1}\}_{i=1}^k\] for $a=1,\ldots,A$. The corresponding junction tree, which we shall refer to as the \emph{slice tree}, is displayed in Figure~\ref{fig:slicetree}. In addition to the cliques and separators arising from the auxiliary variables, the slice tree has $A-1$ cliques each consisting of $4k$ nodes, and $A-2$ separators between them, each consisting of $2k$ nodes. Thus the total size of the slice tree becomes
\[
TS_{\mathrm{slice}} = (A-1)3^{4k} + (A-2) 3^{2k} + TS_{aux}.
\]

\begin{figure}[htb]
  \centering
  \begin{tikzpicture}[node distance = 1em, draw = black!50]
    \foreach \a [evaluate=\a as \an using int(\a+1), evaluate=\a as \ap using int(\a-1)] in {1, ..., 3}{
      \ifthenelse{\a > 1}{
      \clique{C-\a}{Cstyle, right = 3em of C-\ap}{{S_{1\a},S_{1\an},n_{1\a},n_{1\an}},{S_{2\a},S_{2\an},n_{2\a},n_{2\an}},{S_{3\a},S_{3\an},n_{3\a},n_{3\an}}}
        \draw (C-\ap) -- (C-\a);}{
        \clique{C-\a}{Cstyle}{{S_{1\a},S_{1\an},n_{1\a},n_{1\an}},{S_{2\a},S_{2\an},n_{2\a},n_{2\an}},{S_{3\a},S_{3\an},n_{3\a},n_{3\an}}}
}
      \Bclique{B-\a}{Cstyle, below = of C-\a}{{n_{1\a},n_{1\an}},{n_{2\a},n_{2\an}},{n_{3\a},n_{3\an}}}{O_\a}
      \draw (B-\a) -- (C-\a);
      \ifthenelse{\a = 3}{
        \Bclique[1]{B-\an}{Cstyle, right = of B-\a}{{n_{1\an}},{n_{2\an}},{n_{3\an}}}{O_\an}
        \draw (B-\an) -- (B-\a);
      }{}
    }   
  \end{tikzpicture}
  \caption{Slice junction tree for $k=3$ contributors, $A=4$ alleles, and $N=1$ auxiliary variable per allele.}
  \label{fig:slicetree}
\end{figure}
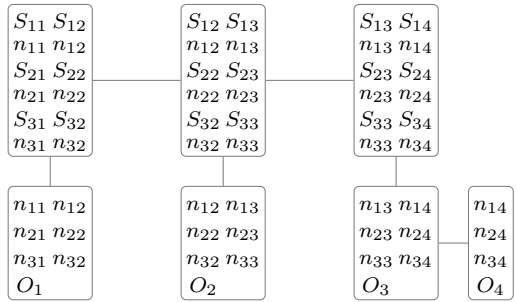
 
However, we can improve on this triangulation by splitting each slice into two cliques as Figure~\ref{fig:initsplit} illustrates.
\begin{figure}[htb]\begin{center}
\begin{tikzpicture}[draw = black!50]
  \clique{C-1}{Cstyle}{{S_{1,a},S_{1,a+1},n_{1,a},n_{1,a+1}}, {S_{2,a},S_{2, a+1},n_{2,a},n_{2,a+1}}, {S_{3,a},S_{3,a+1},n_{3,a},n_{3,a+1}}}
  \clique{C-2a}{Cstyle, right = 3em of C-1}{{S_{1,a},,n_{1,a},n_{1,a+1}}, {S_{2,a},,n_{2,a},n_{2,a+1}}, {S_{3,a},,n_{3,a},n_{3,a+1}}}
  \clique{C-2b}{Cstyle, right = 1em of C-2a}{{S_{1,a},S_{1,a+1},,n_{1,a+1}},{S_{2,a},S_{2,a+1},,n_{2,a+1}},{S_{3,a},S_{3,a+1},,n_{3,a+1}}}
   \draw[|->, draw = black] ( $ (C-1.east)!.25!(C-2a.west) $ ) -- ( $ (C-1.east)!.75!(C-2a.west) $ );
  \draw (C-2a) -- (C-2b);
\end{tikzpicture}
\end{center}
\caption{\label{fig:initsplit} Splitting each slice into two cliques consisting of lower and upper for a reduction in total size.}
\end{figure}
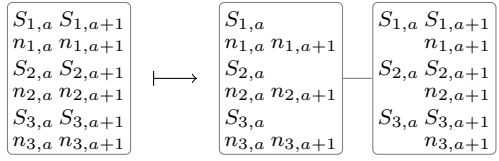
The resulting  \emph{triangle tree} has $2(A-1)$ cliques of each $3k$ nodes and $2(A-1)$ separators of each $2k$ nodes, and thus the total size
\[
TS_{\mathrm{triangle}} = 2(A-1)3^{3k} + \{2(A-1)-1\}3^{2k} + TS_{\mathrm{aux}}
\]
grows significantly slower with the number of unknown contributors than the slice tree; see Figure~\ref{fig:totalsize}.

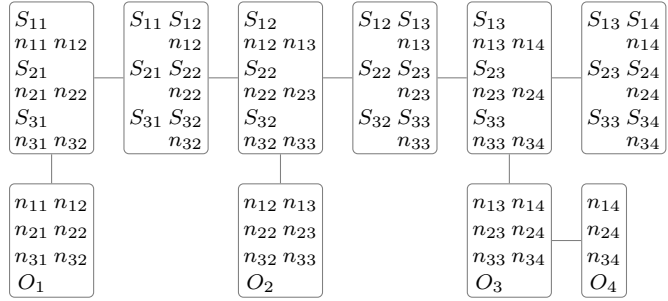
\begin{figure}[htb]
  \centering
  \begin{tikzpicture}[node distance = 1em, draw = black!50]
    \foreach \a [evaluate=\a as \an using int(\a+1), evaluate=\a as \ap using int(\a-1)] in {1, ..., 3}{
      \ifthenelse{\a = 1}{
        \clique{D-1-\a}{Cstyle}{{S_{1\a},,n_{1\a},n_{1\an}},{S_{2\a},,n_{2\a},n_{2\an}},{S_{3\a},,n_{3\a},n_{3\an}}}
      }
      {
        \clique{D-1-\a}{Cstyle, right = of D-2-\ap}{{S_{1\a},,n_{1\a},n_{1\an}},{S_{2\a},,n_{2\a},n_{2\an}}, {S_{3\a},,n_{3\a},n_{3\an}}}  
        \draw (D-2-\ap) -- (D-1-\a);
      }
      \clique{D-2-\a}{Cstyle, right = of D-1-\a}{{S_{1\a},S_{1\an},,n_{1\an}},{S_{2\a},S_{2\an},,n_{2\an}},{S_{3\a},S_{3\an},,n_{3\an}}}
      \draw (D-1-\a) -- (D-2-\a);
      \Bclique{B1-\a}{Cstyle, below = of D-1-\a}{{n_{1\a},n_{1\an}},{n_{2\a},n_{2\an}},{n_{3\a},n_{3\an}}}{O_\a}
      \draw (B1-\a) -- (D-1-\a);
      \ifthenelse{\a = 3}{
        \Bclique[1]{B1-\an}{Cstyle, right = of B1-\a}{{n_{1\an}},{n_{2\an}},{n_{3\an}}}{O_\an}
        \draw (B1-\an) -- (B1-\a);
      }{}
    }
  \end{tikzpicture}
  \caption{Triangle junction tree for $k=3$ contributors, $A=4$ alleles, and $N=1$ auxiliary variable per allele.}
  \label{fig:triangletree}
\end{figure}

In the case of only one unknown contributor, the total size of the triangle tree
cannot be reduced. However, with more than one unknown contributor, each
clique containing $k$ upper triangles can be further split into $k$
cliques as in Figure~\ref{fig:split}.
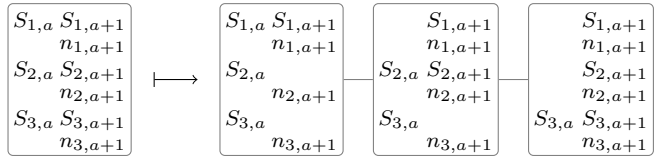
\begin{figure}[htb]
\begin{center}
\begin{tikzpicture}[draw = black!50]
  \clique{C-1}{Cstyle}{{S_{1,a},S_{1,a+1},,n_{1,a+1}},{S_{2,a},S_{2,a+1},,n_{2,a+1}},{S_{3,a},S_{3,a+1},,n_{3,a+1}}}
  \clique{C-2a}{Cstyle, right = 3em of C-1}{{S_{1,a},S_{1,a+1},,n_{1,a+1}},{S_{2,a},,,n_{2,a+1}},{S_{3,a},,,n_{3,a+1}}}
  \clique{C-2b}{Cstyle, right = 1em of C-2a}{{,S_{1,a+1},,n_{1,a+1}},{S_{2,a},S_{2,a+1},,n_{2,a+1}},{S_{3,a},,,n_{3,a+1}}}
  \clique{C-2c}{Cstyle, right = 1em of C-2b}{{,S_{1,a+1},,n_{1,a+1}},{,S_{2,a+1},,n_{2,a+1}},{S_{3,a},S_{3,a+1},,n_{3,a+1}}}
  \draw[|->, draw = black] ( $ (C-1.east)!.25!(C-2a.west) $ ) -- ( $ (C-1.east)!.75!(C-2a.west) $ );
  \draw (C-2a) -- (C-2b);
  \draw (C-2b) -- (C-2c);
\end{tikzpicture}
\end{center}
\caption{\label{fig:split} Splitting upper triangle cliques for a further reduction in total size.}
\end{figure}
Note that the cliques containing $k$ lower triangle sets cannot be split in a similar fashion. 
The resulting junction tree then has $A-1$ cliques of each $3k$ nodes, a further $k(A-1)$ of each $2k+1$ nodes, and $(k+1)(A-1)-1$ separators of $2k$ nodes between them. The total size of the tree is thus
\[ TS_{\mathrm{opt}} = (A-1)3^{3k} + \{(4k+1)(A-1)-1\}3^{2k} +
TS_{\mathrm{aux}}.
\]
 
A further slight reduction of the total size can be obtained by a small alteration in the cliques that cover nodes from the first two and last three alleles; the resulting tree is seen in Figure~\ref{fig:optjts}. We shall refer to this tree as the  \emph{optimal} tree, as this is the best junction tree we have been able to construct. We have also investigated junction trees found by using triangulation algorithms implemented in \hugin but none have smaller total size than our optimal tree. 

\begin{figure}[htb]
  \centering
  \includegraphics[width = \textwidth]{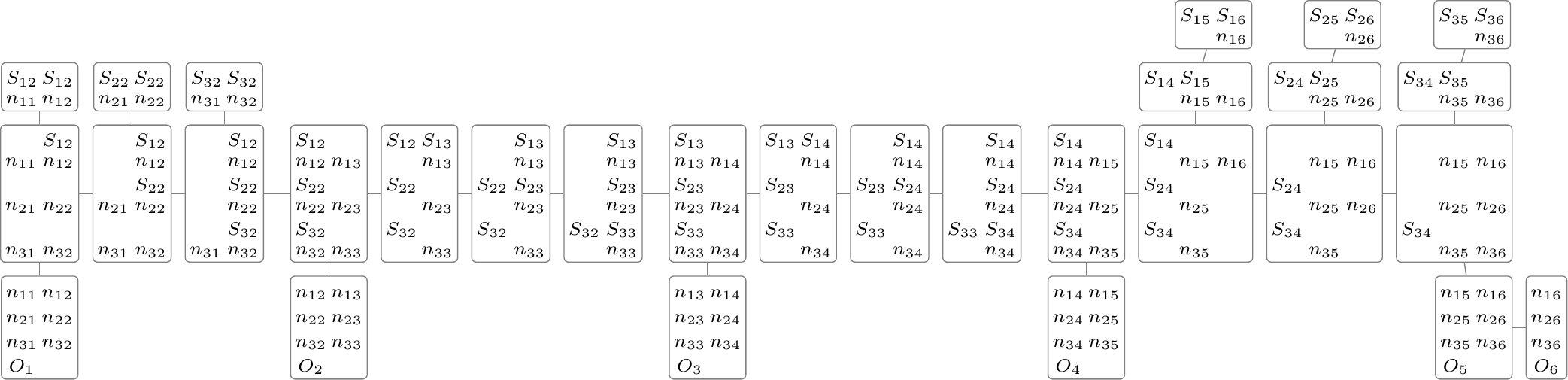}
  \caption{Optimal junction tree for a DNA mixture network with $k=3$, $A=6$, and $N=1$.}
  \label{fig:optjts}
\end{figure}

The optimal junction tree can be generated by an elimination sequence which first eliminates all the auxiliary variables and then proceeds
through the network nodes as
\[
\bS_{A}, \bS_{A-1}, \bS_{1}, \bn_{1}, \left\{\bn_{a}, \bS_a\right\}_{a = 2}^{A-2}, \bn_{A-1}, \bn_{A}
\]
where $\bS_a$ denotes $\{S_{ia}\}_{i=1}^k$ etc. %

The exponential growth of the total size of the three types of junction tree is illustrated in Figure~\ref{fig:totalsize}. Our numerical examples all include $N=3$ auxiliary variables for
each allele to reflect the size of the networks used in the R-package \texttt{DNAmixtures}. The choice of $N$ makes little difference to the total size as this in all cases grows linearly with $N$.

The network representations constructed for the genotypes have a large number of state combinations that are impossible, for example due to the constraint that $\sum_a n_{ia}=2$ for all $i$. In \hugin there is a facility to \emph{compress} the domain, such that
only configurations of clique and separator states with non-zero probability are stored, thus reducing the effective size of the junction tree. There is a slight cost in terms of
book-keeping, but for our purposes this cost is negligible.

\begin{figure}[htb]
  \centering
  \includegraphics{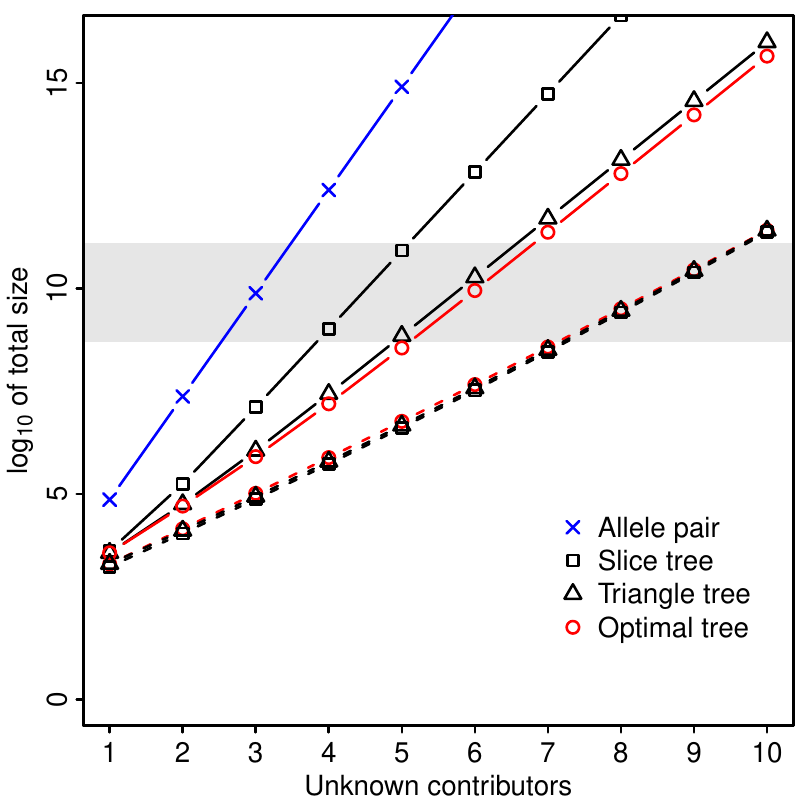}
  \caption{\label{fig:totalsize} Total sizes of junction trees
    as a function of the number $k$ of unknown contributors, in the
    case of $A=25$ allelic types and $N=3$ auxiliary variables per
    allele. Solid lines are uncompressed sizes and dashed lines
    compressed sizes. The horisontal band indicates total sizes ranging from 2GB to 512GB assuming numbers are represented in single precision.}
\end{figure}

As is apparent from Figure~\ref{fig:totalsize}, the exponential growth pattern prevails for the compressed domains. Note that after compression all three junction trees are approximately of the same size. Also, the reduction of total size obtained by compression is itself growing exponentially;
ignoring any slight reduction in total size 
from compressing states with probability zero in the cliques with auxiliary variables, the total size for the compressed slice tree is
\[
TS_{\mathrm{compr.slice}} = (A-3)10^k + \left\{3N(A-1)+A\right\}6^k + 3N 3^k.
\]

 In general, to make a compression, one single propagation has to be performed and therefore the uncompressed networks set the limit for computational feasibility. When numbers are represented in single precision of each four bytes, the horisontal band in Figure~\ref{fig:totalsize} represents a range of  capacities from 2GB to 512 GB of memory.

Figure~\ref{fig:totalsize} indicates that using the optimal junction tree should enable computation for up to $k=6$ unknown contributors, whereas using the slice tree restricts computation to around $k=4$.  

There is a simple way of compressing the slice tree in that there are at most 10 possible configurations of the states in each of $\{S_{ia}, S_{i,a+1}, n_{ia} , n_{i,a+1}\}$. So if the state space is defined by these from the outset, it would in principle be possible to handle up to $k=9$ unknown contributors, as it the compressed network would determine the maximal capacity; however, the general flexibility of the representation would be reduced.

\subsubsection{Other representations of genotypes}

Clearly, the network that represents the genotype of an unknown
contributor could be replaced by a different representation than the one suggested here and connected to the auxiliary variables in an appropriate way. We shall briefly consider two alternative representations of a genotype.

\paragraph{Allele-pair representation} More commonly, a genotype has been represented directly as an unordered pair of alleles; this representation has for example been used in \cite{Cowell2011202}. Including $A$ alleles in the model there are
$A(A+1)/2$ possible unordered pairs. If an allele-pair is represented by a single node for each contributor, the parent set for each auxiliary variable  in this network is the collection of the $k$ unknown genotype-nodes, resulting in a junction tree where each clique and each separator contains all of the $k$ genotype-nodes. 
Adding $N$ auxiliary variables for each of $A$ alleles yields the total size  
\[
TS_{\textrm{allele-pair}} = (3NA-1)\left\{A(A+1)/2\right\}^k.
\] 

We note that this junction tree exhibits polynomial rather than linear growth in $A$, rendering the representation less efficient for markers with a large number of possible allelic types. For a fixed number of alleles, the growth in the number $k$ of unknown contributors is still exponential; see Figure~\ref{fig:totalsize}. For junction trees based on the Markov representation of genotypes, the number of alleles makes a neglible impact on the total size. However, for the allele-pair representation the rate of growth depends heavily on the number of allelic types: For 25 alleles as in Figure~\ref{fig:totalsize} it is feasible to handle up to about 3 unknown contributors, whereas if only 10 allelic types are needed, then 4-5 unknown contributors can be handled. For 7 or more allelic types,  the Markov representation in combination with optimal triangulation is superior to the allele-pair representation regardless of the number of unknown contributors. 
As the allele-pair representation is compressed by construction, there is no possibility of further compression of the junction tree.

\paragraph{Single gene representation} Another possibility, used for example in \cite{dawid:etal:02} and \cite{mortera:etal:03}, is to model the genotype at the single gene level. A single gene can
be represented by the same Markovian network structure as that in Figure~\ref{fig:onegt} used for a genotype, just that each node $n_{ia}$ or $S_{ia}$ has state space $\{0,1\}$ rather than
$\{0,1,2\}$. However, there is a cost in that two such networks are needed per unknown contributor, 
resulting in a total size with growth-rate $O(A\times 2^{3(2k)})$ compared to $O(A\times3^{3k})$ when using the genotype representation. Thus, the single gene network will always be inferior to the genotype network.

The total size of the optimal single gene tree renders computations  feasible for up to about 5 unknown contributors. 
 Compression of the single gene slice tree yields a growth rate of $O(A\times 16^k)$, which still is considerably higher than $O(A \times 10^k)$ for the corresponding compressed genotype slice tree. It would stay feasible if $k\leq 7$.

For $A\ge 11$ allelic types, the single gene representation compares favourably to the allele-pair representation.

Although inefficient, the single gene network representation may be preferable for other reasons; for example in cases where the two genes might be selected from different populations, if sensitivity to uncertainty or population structure should be investigated as in \cite{green:mortera:09}, or if there is additional complexity involving family relations etc.\ as in \cite{mortera:etal:03}. 
\section{DNA mixture analysis}
\label{sec:mixtureanal}

The analysis of a mixed trace  can have different objectives  depending on the context. The objective can be a quantification of the strength of
\emph{evidence} for a given hypothesis over another, or the objective may be a \emph{deconvolution} of the trace, i.e.\  that one wishes to predict 
genotypes of unknown contributors. 

As a generic example we consider a trace \emph{MC15} from  \cite{Gill200891}, also analysed in \cite{cowell:etal:13}. The trace is believed to contain DNA from at least three contributors, and the victim, who we shall denote $K_1$, is assumed present along with another contributor $K_2$. We shall here deal with the question of the identity of the third contributor. 
 The peak heights from one marker are given in Table~\ref{tab:mc15dat} along with the allele-counts for each of three genotyped individuals.

The available \emph{evidence} $E$  consists of the peak heights as observed in the EPG as well as the genotypes of individuals associated with the case. It is customary to assume relevant population gene frequencies to be known. 

\begin{table}[htb]
	\caption{Peak heights for marker D2S1338 above threshold in trace MC15, and genotypes of associated individuals.} 
	  \centering\small
\begin{tabular}{lrrrrr}
 Allele & Peak height & \multicolumn{3}{c}{Allele-count}\\
$a$ & $Z_a$ & $K_1$& $K_2$ & $K_3$\\
 \hline
   16 & 64 & 0 &0 & 1   \\
       17 &  96 & 0 &0 & 1      \\
      23  & 507 & 1 & 0&0 \\
       24 &  524 & 1 &2 & 0\\
     \hline
\end{tabular}
\label{tab:mc15dat}
\end{table}

\paragraph{Strength of evidence.}

We now consider two competing explanations to the trace.

The \emph{prosecution hypothesis} $H_p: K_1\&K_2\&K_3$ claims that the trace has exactly three contributors who are identical to the three known individuals $K_1$, $K_2$, and $K_3$.

An alternative explanation of the trace is the \emph{defence hypothesis} $H_d: K_1\&K_2\&U$ that the trace contains the DNA of $K_1$, $K_2$, as well as that of an unknown and unrelated individual $U$, whereas $K_3$ has not contributed. 

The strength of the evidence is reported as a \emph{likelihood ratio}:
\[LR={L(\hat H_p)}/{L(\hat H_d)}={\Pr( E \cd \hat H_p)}/{\Pr(E \cd \hat H_d)}\]
where $\hat H_i$ indicates that we use the maximum likelihood estimates of the parameters under the hypothesis $H_i$, see Table~\ref{tab:mle15} below.

\paragraph{Deconvolution.} Under the defence hypothesis we are interested in determining the identity of the unknown contributor $U$. This could for example be done by finding the most probable genotypes for $U$ given the evidence, i.e.\ those with the  highest values of $\Pr( U \cd \hat H_d, E)$. We shall return to this issue in Section~\ref{sec:map} below.

\paragraph{Estimation.}
In order to calculate the relevant quantities for any of the above questions, we need to estimate the unknown parameters of the model. Being able to evaluate the likelihood function, this can be done by numerical maximisation. The maximum likelihood estimates and standard errors obtained under the defence hypothesis $H_d$ and prosecution
hypothesis $H_p$ are given in Table \ref{tab:mle15}. The resulting likelihood ratio is $\log_{10}( LR)= 12.12$. 
\begin{table}[htb]
\caption{Maximum likelihood estimates based on MC15.}
\centering \small
\begin{tabular}{lr|lr}
\multicolumn{2}{c|}{Defence hypothesis}&\multicolumn{2}{c}{Prosecution hypothesis}\\
Parameter&Estimate &Parameter&Estimate\\
\hline
$\rho$   &       26.95            & $\rho$ &    33.86 \\
$\eta $ &      33.86    &      $\eta$   &   26.94  \\
$\xi$      &  0.086       &    $\xi$      &  0.076  \\
$\phi_{K_1}$ &  0.823   &         $\phi_{K_1}$ & 0.825 \\
$\phi_{K_2}$ & 0.055   &         $\phi_{K_2}$& 0.049  \\
  $\phi_{U}$&  0.122       &     $\phi_{K_3}$ & 0.126  \\
\hline
$\log_{10}L(\hat H)$ &{-130.21}&$\log_{10}L(\hat H)$& {-118.09}\\
\hline
\end{tabular}
\label{tab:mle15}
\end{table}

\subsection{Model Diagnostics}\label{sec:diagnostics}

In the assessment of forensic evidence, little attention has been devoted to demonstrate the adequacy of a proposed model used to analyse a specific case or, of equal importance, to assert that data have been correctly recorded for the analysis. This may partly be due to the unavailability of useful methods for the purpose. However, we believe
this aspect to be of utmost importance; in particular we find it reasonable that one should not only compare the prosecution and defence hypothesis, but there should also be an effort to demonstrate that neither hypothesis represents an implausible explanation of the trace under analysis. 

Previously we have introduced auxiliary variables $O_a$, to enable simple computation of the likelihood function  \eqref{eq:likfunction} and representation of evidence from observed peak heights \eqref{eq:evidence}. We shall in the following introduce further auxiliary variables such as binary variables $D_a$ which indicate whether or not a peak was observed for allele $a$, and variables $Q_a$ which indicate whether a peak observed at allele $a$ was less than a specified value. Both of these types of auxiliary variables shall prove to be useful for model validation; in addition, the variables $D_a$ can be used in an analysis which refrains from exploiting the peak heights but is based only on peak presence; see Section~\ref{sec:peakpresence} below.

\subsubsection{Assessing peak height distributions}

First, we wish to investigate whether our model appropriately predicts the observed peak heights. 
Given $Z_a \ge C$, the
peak height follows a continuous distribution and thus the probability transform 
$\Pro(Z_a \le z_a \cd Z_a \ge C)$  follows a uniform distribution.

To express the probability in a way suitable for computation with auxiliary variables we first note that for $z \ge C$ we have
\[
\Pro(Z_a \le z \cd Z_a \ge C) = \frac{\Pro(Z_a \le z) - \Pro(Z_a < C)}{\Pro(Z_a \ge C)}.
\]
Thus all we need to evaluate is the distribution
function in the observed value $z_a$ and at the threshold
$C$. The distribution function
\begin{equation}\label{eq:distfun}
\Pro(Z_a \le z) = \Exp \Big\{\Pro(Z_a \le z \cd \bn_a, \bn_{a+1})\Big\}
\end{equation}
is the expectation of a trivial product of one factor, and to compute this we add
an auxiliary variable $Q_a$ with the same parents as for $O_a$ and with
conditional probability
\[
\Pro(Q_a = 1 \cd \bn_a, \bn_{a+1}) = \Pro(Z_a \le z \cd \bn_a, \bn_{a+1}).
\]
Similarly, we add a binary variable $D_a$ allowing the evaluation of both $\Pro(Z_a \ge C)$ and
$\Pro(Z_a < C)$. 

It can be of interest to consider the distribution of the peak height in
the light of other observed peaks, and not just the marginal
distribution of the peak itself. For instance, we can condition on the
peak heights of all other alleles to get $\Pro\cdp{Z_a \le z}{Z_b =
  z_b, b \neq a, Z_a \ge C}$, or we could include this information for
only the preceding alleles in the ordering to get $\Pro\cdp{Z_a \le
  z}{Z_b = z_b, b \le a, Z_a \ge C}$.  These distributions can all be obtained simply
through conditioning on relevant variables $O_a$ as described in
Section~\ref{sec:posterior}. 

In Figure~\ref{fig:qqplots}, quantile-quantile plots for the conditional distribution of a peak height given observed peak heights  for all other alleles are shown for $H_p$ and $H_d$ using trace MC15 and the associated maximum likelihood estimates in Table~\ref{tab:mle15}.
\begin{figure}[htb]
  \centering
  \includegraphics{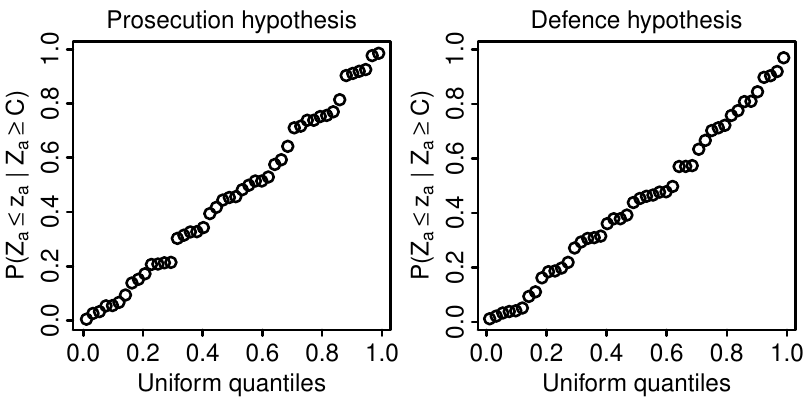}
  \caption{Quantile-quantile plots for the prosecution and defence hypotheses for MC15.}
  \label{fig:qqplots}
\end{figure}

We note that in both diagrams the points are close to the identity line and there is no indication that the peak height distributions are inadequately modelled  under either of the hypotheses.

We can also take a closer look at the distribution of the peak height at any single allele, for example to identify outlying observations. This is illustrated in Figure~\ref{fig:boxplots}. Boxes indicate quartiles  and whiskers indicate 0.5\% and 99.5\% prediction limits for the conditional distributions of peak heights $\Pro\cdp{Z_a \le z}{Z_b =
  z_b, b \neq a, Z_a \ge C}$. The quantiles are found by numerical inversion of the distribution function \eqref{eq:distfun}.
  \begin{figure}[htb]
  \centering
 \includegraphics{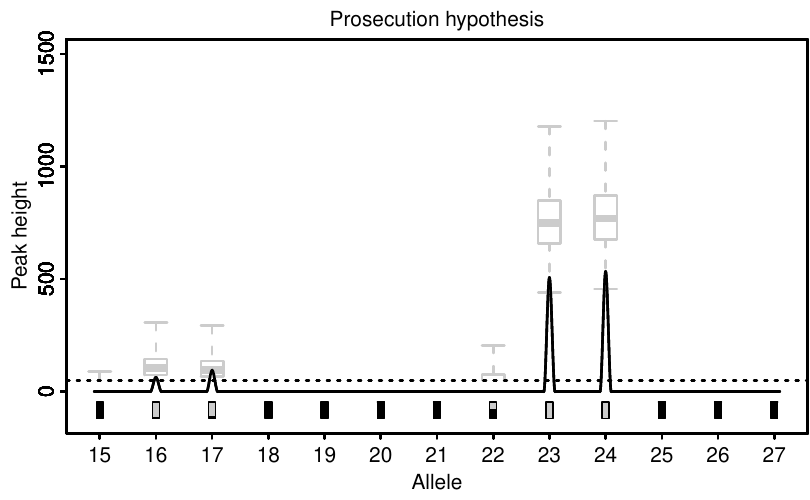}
 \includegraphics{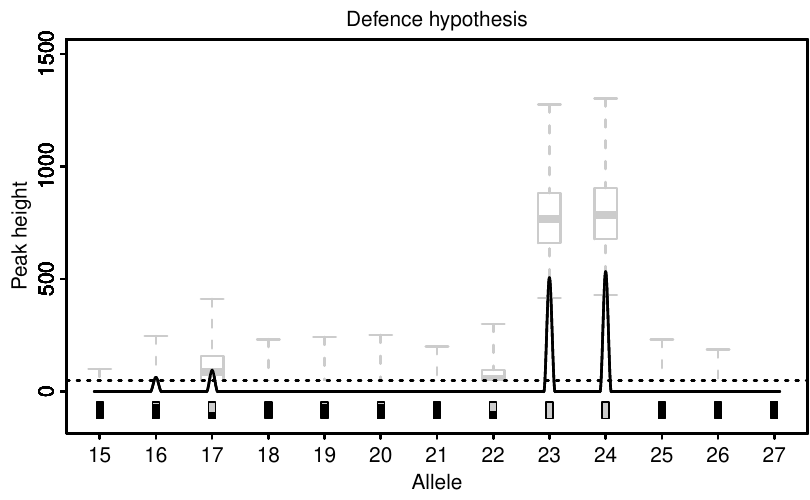}
 \caption{Comparison of observed peak heights to their predictive distribution conditionally on all other observed peak
   heights for marker D2S1338.  The bar below each peak indicates the
   probabilities of observing (grey) and not observing (black) a peak
   at this allele.}
  \label{fig:boxplots}
\end{figure}

 We note that although the observed peak heights at alleles 23 and 24 are somewhat lower than expected, there are no observations that are clear outliers, conforming with the quantile-quantile plots in Figure~\ref{fig:qqplots}. Note that the prosecution hypothesis predicts complete absence of peaks at alleles 18--21 and 25-27, whereas this is not the case for the defence hypothesis involving alleles from unknown contributors; hence under this hypothesis peaks are \emph{a priori} possible at any allele.

\subsubsection{Prequential monitoring of peak presence}

Next, we wish to investigate whether our model correctly predicts
absence and presence of peaks in the EPG. We use the prequential theory of \cite{dawid:84} with so-called
prequential
monitors \citep{article:forecastvalidity}.

Using some arbitrary ordering, we consider the set of alleles across all
markers and the probability that a peak has been seen for allele $a$
given the peak heights observed on all preceding alleles,
\[
p_a = \Pro\cdp{Z_a \ge C}{z_i, i < a}=\Pro\cdp{D_a =1}{z_i, i < a}
\]
which can be obtained by propagation as described in Section~\ref{sec:posterior}.
For each allele $a$, we then consider the logarithmic score
\[
Y_a = 
\begin{cases}
  -\log p_a, & \mbox{if } z_a \ge C\\
  -\log (1-p_a), & \mbox{if } z_a < C
\end{cases}
\]
so that $Y_a$ is always non-negative and higher values of $Y_a$ represent a large penalty for assigning a small probability ($p_a$ or $1-p_a$) to the event that actually happens. 

The cumulative logarithmic score, adjusted for incremental expectations,
\[
M_a = \sum_{i = 1}^a \left\{Y_i - \Exp\cdp*{Y_i}{Z_b, b<i} \right\}
\]
is a martingale with respect to the sequence of peak heights. 

As
$\Var\cdp*{M_a - M_{a-1}}{Z_b, b<a} = \Var\cdp*{Y_a}{Z_b, b<a}$,
the distribution of the normalised cumulative score
\[
\frac{\sum_{i=1}^a Y_i - \sum_{i=1}^a\Exp\cdp*{Y_i}{Z_b, b<i}}{\sqrt{\sum_{i = 1}^a \Var\cdp*{Y_i}{Z_b, b<i}}}
\]
approaches a standard normal distribution as the denominator becomes infinitely large \citep{article:forecastvalidity}. Thus for $q_{1-\alpha}$ being the $1-\alpha$ quantile of the standard normal distribution,
\[
q_{1-\alpha}\;\sqrt{\sum_{i = 1}^a \Var\cdp*{Y_i}{Z_b, b<i}}
\]
is an approximate pointwise $1-\alpha$ upper predictive limit for the cumulative score at allele $a$. 

The cumulative score can easily be calculated using that if $p_a\in \{0,1\}$ we have $Y_a=0$ and otherwise
\begin{align*}
  \Exp\cdp*{Y_a}{Z_b, b < a} &= -p_a\log p_a - (1-p_a)\log(1-p_a),\\
  \Var\cdp*{Y_a}{Z_b, b < a} &= p_a(1-p_a)\left\{\log p_a- \log(1-p_a)\right\}^2.
\end{align*}

Prequential monitor plots of the  prosecution and defence hypothesis for MC15 are displayed in Figure~\ref{fig:preqplots}. 
\begin{figure}[htb]
  \centering
  \includegraphics{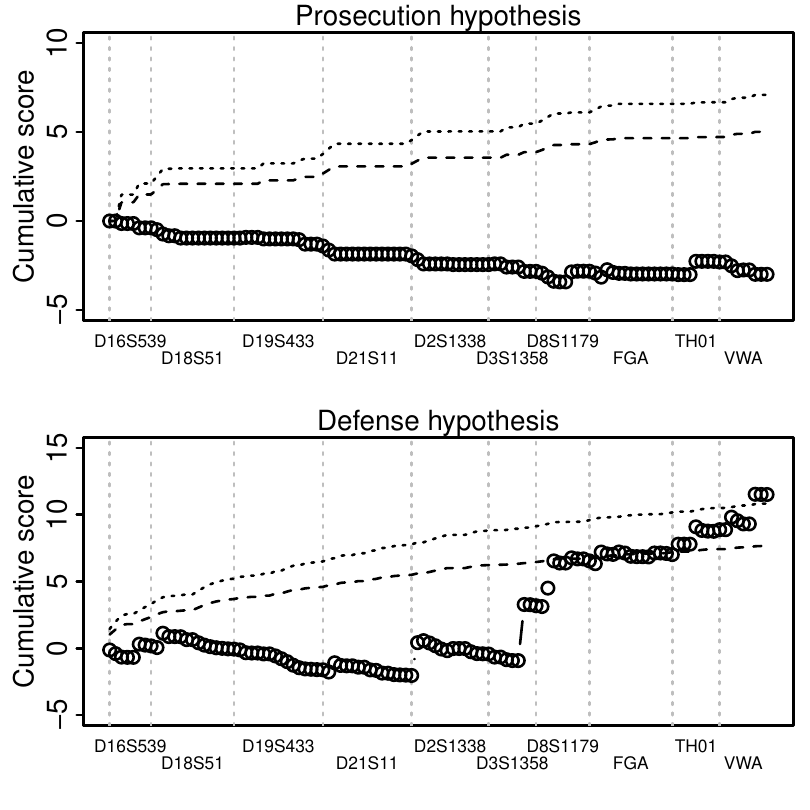}
  \caption{Prequential monitor plots of the prosecution and defence hypotheses for MC15. The dashed horisontal lines indicate upper 95\% and 99\% pointwise predictive limits based on the approximating normal distribution.}
  \label{fig:preqplots}
\end{figure}

A negative jump in the score means that we have observed what the
model predicts as most likely, whereas a positive jump means that we
have observed the opposite of what is most likely according to the
model. If it is equally likely for a peak to fall above and below the
threshold, or there is only one possible outcome --- i.e.\ if $p_a\in
\{0, 1/2,1\}$ --- there is no jump. The size of an upward jump
indicates the level of disagreement between model and observations.
Note that for the defence hypothesis, the
monitors cross the upper limits towards the end of the plot,
indicating that this hypothesis may not adequately describe the
pattern of observed peaks. Further investigation may reveal whether upward jumps are due to observation of rare alleles or, for example, due to recording errors in the data.

\subsection{Simulation}
\label{sec:simulation}

As noted in Section~\ref{sec:posterior}, introducing evidence on the
auxiliary variables $O_a$ yields a representation of the posterior
distribution of the genotypes of the unknown contributors. This in
turn enables simulation of a full DNA trace including peak heights, either 
marginally or conditionally on relevant subsets of the observed peak
heights. More generally, we have for any event $B$ that
\[f_\psi\cdp*{\{z_a\}_{a\in A},\bn}{B} =f_\psi\cdp*{\{z_a\}_{a\in A}}{\bn,B}p\cdp*{\bn}{B}.\] 
If conditioning with $B$ can be
represented by propagation in our Bayesian network, for example if $B=
\{Z_b = z_b, b \neq a\}$, we can easily simulate from $p\cdp*{\bn}{B}$
by standard methods \citep[Section 6.4.3]{cowell:etal:99}. Thus to sample a
full DNA trace, we just further need a method for sampling from
$f_\psi\cdp*{\{z_a\}_{a\in A}}{\bn,B}$.
  
This method of simulation can for example be used in a bootstrap analysis of the estimation uncertainty as in \cite{graversen:lauritzen:13}. Simulation could also be relevant for assessing the discriminatory ability of the calculated likelihood ratio, for illustration of peak height variability, and other forms of model validation. Below we are exploiting simulation in the prediction of profiles of unknown contributors.

\subsection{Prediction of unknown profiles}
\label{sec:map}

In a model involving unknown contributors it can be relevant to
investigate the distribution of  genotypes for each of these conditionally on the evidence. Focusing on a single or few alleles, we can explore 
  this distribution directly. For any combination of genotypes we can compute its probability exactly by probability propagation. We can 
identify those of highest probability by sampling 
genotypes until a proportion $p$ of
the probability mass has been visited as then each of the remaining
combinations of genotypes must have probability at most
$1-p$. Thus the $r$  combinations with probability
strictly greater than $1-p$ must be among those sampled. They can then be ranked according to their
probability and constitute the list of the $r$ most probable
combinations. Here the number $r$ depends on the probability $p$
chosen.

Considering the defence hypothesis of trace MC15, we would like to identify the genotype of the unknown contributor $U$. If we consider the full genotype, at all markers, we often get a very diffuse distribution  as for example reported in \cite{cowell:etal:13}.

One reason for this is that, due to dropout, there are generally many unseen alleles that could be present in the mixture without giving rise to a peak. 
However, if we focus on  explaining the peaks actually seen in the EPG we get a more concentrated distribution, as displayed in  Table~\ref{tab:map}, where the total probability of the six combinations add up to one. 
\begin{table}[htb]
  \centering
  \caption{Probabilities of genotype at marker D2S1338 for the unknown contributor $U$ under the defence hypothesis. The defendant $K_3$ has genotype (16,17).}
  \label{tab:map}
  \begin{tabular}{cccccc}
    16 & 17 & 23 & 24 & D & Prob \\ 
    \hline
    1 &     1 &     0 &     0 &     0 & 0.5276 \\ 
    0 &     1 &     0 &     0 &     1 & 0.1861 \\ 
    0 &     2 &     0 &     0 &     0 & 0.1697 \\ 
    0 &     1 &     0 &     1 &     0 & 0.0640 \\ 
    0 &     1 &     1 &     0 &     0 & 0.0509 \\ 
    1 &     0 &     0 &     0 &     1 & 0.0017 \\ 
    \hline
    \multicolumn{5}{l}{Total probability} & 1.0000 \\ 
    \hline
  \end{tabular} 
\end{table}
As the table shows, the probability that the unknown contributor has at least one allele 17 is .9983, close to certainty. There is some uncertainty concerning the second allele which can be virtually anything although it is by far most probable that
the  genotype is (16, 17); this genotype is that of the defendant $K_3$. The second most probable explanation of the trace is that the other allele has dropped out.

\subsection{Strength of evidence when ignoring peak heights}
\label{sec:peakpresence}

Another potential application of the auxiliary variables is to calculate a likelihood ratio which only uses information about peak presence or absence. This can be done by specifying evidence for the nodes $D_a$ introduced in Section~\ref{sec:diagnostics} rather than for  nodes $O_a$.

It is still necessary to specify a set of model parameters, which for example could be estimated using peak heights. Using the estimates in Table~\ref{tab:mle15} we obtain a likelihood ratio of $\log_{10} LR = 9.85$ which is weaker than the evidence obtained with full peak height information but it is still incriminating for the defendant. Such an analysis is analogous to the one used in \texttt{likeLTD} as suggested by \cite{balding:13}, where peak heights are used only to classify peaks as present, absent, or uncertain.

We have used peak heights to estimate the parameters of the model. In principle parameters could also be estimated solely on the peak presence information, possibly in combination with prior information on some of these, although such estimates would be ill-determined and therefore not useful. 

\subsection{Multiple mixed traces}

By adding more auxiliary variables to the model, we can easily extend the
model to handle multiple traces, either with independent unknown
contributors or where some or all unknown contributors coincide.

We assume that the peak heights across mixed traces are
conditionally independent given the genotypes of common contributors. Peak
height distributions are allowed to vary across traces through the model
parameters.

The network now models the set of all unknown contributors to the mixed traces.
Denote by $\phi^j_i$ the proportion of DNA that contributor $i$ has
made to trace $j$. Then $\phi^j_i = 0$ corresponds to contributor
$i$ not being present in trace $j$.  Therefore, the case where some
or all contributors are distinct to a particular mixed trace is a
sub-model corresponding to $\phi^j_i = 0$ for some $(i,j)$.

An advantage of this specification of the joint model is that we do not need to make
assumptions about possible common unknown contributors to the traces, but we can
let the maximisation of the likelihood point to the relevant
scenario. This has been used in \cite{cowell:etal:13} for a combined analysis of MC15 with another trace pertaining to the same case.

In the case where the
traces have completely independent unknown contributors, it is
recommendable to represent each trace as a separate network to limit the
number of unknown contributors in each network.

\section{Discussion}
We note that our computational methods are exact throughout under the model adopted, and that the only approximations relate to the model representing an inevitable approximation to reality, and possible imprecision of numerical methods.  Nevertheless, using the efficient junction tree representations and exact compression methods as described in Section~\ref{sec:complexity}, we are able to handle more contributors than what has previously been possible. 

We have far from exhausted the flexibility and the potential of the Bayesian network model and point out that simple modifications or elaborations of the basic network can readily be used to, say, incorporate the presence of silent alleles simply by including an extra allele in the genotype representation, or to enable the direct computation of the probability that a specific peak is due to stutter or an absent peak is due to random dropout or allele absence; see \cite{cowell:etal:13} for this and further examples.

\end{document}